\documentclass[journal]{IEEEtran}

\usepackage{amsmath,amsfonts,amssymb}
\usepackage{bbm}
\usepackage{caption}
\usepackage{subcaption}
\usepackage{comment}
\usepackage{graphicx}
\usepackage{algorithm}
\usepackage{algorithmic}
\usepackage{amsthm}
\usepackage{orcidlink}
\usepackage{hyperref}
\hypersetup{
    colorlinks=true,
    linkcolor=teal,
    filecolor=magenta,      
    urlcolor=cyan,
    citecolor=teal
    }
\usepackage{cite}
\usepackage{tikz}
\usepackage{url}
\usepackage{adjustbox}
\usepackage{balance}
\usetikzlibrary{arrows.meta, positioning, calc}
\newtheorem{theorem}{Theorem}

\newtheorem{remark}{Remark}

\begin{document}

\title{ATS-ToDMA: Adaptive Token Selection and Token-Domain Multiple Access for Cross-Modal Semantic Communications}

\author{Sachin~Kadam$^{\orcidlink{0000-0001-7085-3365}}$~and~Dong~In~Kim$^{\orcidlink{0000-0001-7711-8072}}$,~\IEEEmembership{Life Fellow,~IEEE}
\thanks{S.~Kadam is with the Department of Electronics and Communication Engineering, Motilal Nehru National Institute of  Technology, Prayagraj, UP, 211004, India, (e-mail: sachink@mnnit.ac.in) and D.~I.~Kim is with the Department of Electrical and Computer Engineering, Sungkyunkwan University (SKKU), Suwon 16419, Republic of Korea (e-mail: dongin@skku.edu). 
}
}
\maketitle

\begin{abstract}
Adaptive token processing has emerged as a promising approach for improving the efficiency of semantic communication systems. However, existing semantic communication frameworks largely overlook token-level multiple access and the impact of semantic interference among simultaneously transmitted semantic tokens. In this paper, we propose Adaptive Token Selection and Token-Domain Multiple Access (ATS-ToDMA), a novel cross-modal semantic communication framework that jointly performs semantic token selection, interference-aware scheduling, and semantic-aware power allocation. The proposed framework introduces a Semantic Signal-to-Interference-plus-Noise Ratio (SSINR) metric that captures the combined effects of channel impairments and semantic interference arising from token similarity. A transformer-based scheduler is developed to allocate selected semantic tokens across token-domain transmission slots while mitigating both intra-modal and cross-modal semantic interference. To characterize the behavior of the proposed system, analytical bounds on semantic interference and feasible token occupancy are derived, together with a closed-form approximation for semantic-aware power allocation. Simulation results demonstrate significant gains in semantic throughput and semantic decoding accuracy while reducing aggregate semantic interference and transmit power compared with OMA, Semantic NOMA, Random-TS, and Greedy ATS benchmarks. 
\end{abstract}


\section{Introduction}
The rapid proliferation of data-intensive and intelligence-driven applications, such as immersive communications, autonomous systems, and edge intelligence, is fundamentally reshaping the design objectives of wireless networks. Conventional communication systems, grounded in Shannon’s information theory, are primarily designed to ensure reliable bit-level transmission by minimizing distortion between transmitted and reconstructed signals \cite{shannon1948}. However, in many emerging applications, the ultimate goal is not perfect data reconstruction but successful task execution. This mismatch leads to inefficient utilization of communication and computation resources, as large amounts of task-irrelevant information are unnecessarily transmitted.

To address this limitation, semantic communication (SemCom) has recently emerged as a paradigm shift, aiming to transmit only task-relevant information required for inference or decision-making at the receiver \cite{semantic_comm_overview1, semantic_comm_overview2, xin2024semantic, chaccour2025less}. By optimizing task-oriented performance metrics rather than symbol-level accuracy, SemCom enables significant reductions in communication overhead while preserving or even enhancing end-task performance. In particular, deep learning-based approaches, including joint source-channel coding (JSCC), have demonstrated promising results for semantic transmission of multimedia data across noisy wireless channels \cite{task_oriented_comm1, task_oriented_comm2, peng2025robust}.

Recent advances in foundation models and deep representation learning, including transformers for language and vision, have further enabled compact semantic representations across heterogeneous modalities such as text, images, and speech \cite{devlin2019bert, brown2020language, dosovitskiy2020image, radford2021learning}. These models map raw inputs into structured semantic token representations,
where each token encodes localized semantic information. Although such representations significantly reduce redundancy, transmitting all tokens remains inefficient, particularly under multi-user and bandwidth-limited wireless conditions.

To further improve efficiency, adaptive token selection (ATS) techniques have been introduced to retain only the most informative subset of tokens \cite{rao2021dynamicvit, ryoo2021tokenlearner, xu2022evo}. Specifically, ATS selects a subset $\mathcal{K} \subseteq \{1,\dots,N\}$ such that $|\mathcal{K}| \ll N$, while preserving task-relevant semantic content. However, existing ATS methods are primarily designed for computational efficiency and do not explicitly account for wireless resource constraints or multi-user interactions.

More critically, multi-user SemCom introduces a fundamentally new form of interference, termed semantic interference. Unlike conventional electromagnetic interference caused by signal superposition, semantic interference arises from similarity among transmitted semantic tokens. In particular, tokens with high cosine similarity~\cite{lahitani2016cosine}, may introduce ambiguity during decoding, leading to degraded task performance even under favorable channel conditions. This effect becomes more severe in multi-user cross-modal systems, where heterogeneous tasks and modalities coexist. Existing multiple access and scheduling strategies, including orthogonal and non-orthogonal schemes, are not designed to exploit semantic structure, resulting in suboptimal resource utilization~\cite{noma_survey}.

Despite recent advances in semantic communication, three fundamental challenges remain unresolved: (i) identifying and transmitting only semantically valuable tokens in multi-user environments, (ii) mitigating semantic interference arising from similar semantic representations, and (iii) jointly optimizing semantic scheduling and wireless resource allocation across heterogeneous modalities. 

Motivated by these challenges, this paper proposes a unified {adaptive token selection and token-domain multiple access (ATS-ToDMA)} framework for multi-user cross-modal SemCom system. The key idea is to treat semantic tokens as the fundamental transmission units and jointly optimize token selection and scheduling over shared wireless resources. To capture complex semantic dependencies, we develop a transformer-based scheduler based on self-attention mechanisms~\cite{vaswani2017attention}.

In addition, we introduce a semantic interference model that captures both intra-modal and cross-modal interactions among tokens, and we characterize its impact on system performance using a semantic signal-to-interference-plus-noise ratio (SSINR) metric. Based on this model, we derive analytical interference bounds and formulate a semantic throughput maximization problem under reliability and resource constraints.

The main contributions of this paper are summarized as follows:

\begin{itemize}

\item \textbf{Unified ATS-ToDMA Framework:}
We propose a novel Adaptive Token Selection and Token-Domain Multiple Access (ATS-ToDMA) framework for multi-user cross-modal SemCom system. Unlike existing SemCom systems that operate at the signal or feature level, the proposed framework explicitly treats semantic tokens as the fundamental transmission units, enabling fine-grained control over semantic information flow and wireless resource allocation.

\item \textbf{Transformer-Based Semantic Scheduling:}
We design a transformer-based scheduler that captures higher-order semantic dependencies across users and modalities via self-attention mechanisms. This enables joint optimization of token importance and scheduling decisions, significantly improving resource utilization compared to greedy or heuristic allocation strategies.

\item \textbf{Semantic Interference Modeling and SSINR Analysis:}
We introduce a rigorous semantic interference model based on pairwise semantic similarity, capturing both intra-modal and cross-modal interactions among tokens. Based on this model, we define a semantic signal-to-interference-plus-noise ratio (SSINR) metric that quantifies the impact of semantic overlap on task performance.

\item \textbf{Theoretical Performance Characterization:}
We derive closed-form interference bounds and characterize token allocation constraints under SSINR-based reliability requirements. These results provide fundamental insights into the interplay between semantic similarity, interference, and resource allocation in multi-user SemCom systems.

\item \textbf{Semantic Throughput Optimization:}
We formulate a joint optimization problem that maximizes semantic throughput under power, scheduling, and reliability constraints. The problem jointly optimizes adaptive token selection and wireless resource allocation, and we propose an alternating optimization framework to efficiently solve the resulting non-convex problem.

\item \textbf{Complexity and Efficiency Analysis:}
We analyze the computational complexity of LSTM- and transformer-based semantic encoders and demonstrate that the proposed ATS mechanism significantly reduces redundant token processing, leading to lower FLOPs, latency, and communication overhead while preserving semantic fidelity.

\item \textbf{Approximate Closed-Form Power Allocation:} We derive an approximate closed-form semantic-aware power allocation strategy, where the transmit power assigned to each token admits an analytically tractable expression that jointly captures channel noise, semantic interference, and token reliability. The proposed formulation provides both low-complexity implementation and interpretable insights into semantic power adaptation under varying channel conditions.
\end{itemize}
To the best of our knowledge, ATS-ToDMA is among the first semantic communication architectures to explicitly model semantic interference as a schedulable and controllable resource-management quantity. 

Overall, the proposed framework provides a unified theoretical and algorithmic foundation for scalable, interference-aware, and resource-efficient SemCom in next-generation wireless networks. The limitations of existing state-of-the-art approaches compared to the proposed ATS-ToDMA framework are provided in Table~\ref{tab:limitations}.

\textbf{Paper Organization:}
The remainder of this paper is organized as follows. Section II reviews related work. Section III presents the system model and problem formulation. Section IV develops the proposed ATS-ToDMA framework. Section V provides theoretical analysis and complexity results. Section VI presents simulation results. Finally, Section VII concludes the paper and outlines future research directions.
\begin{table}
\centering
\caption{Limitations of existing approaches compared to the proposed ATS-ToDMA framework.}
\label{tab:limitations}
\renewcommand{\arraystretch}{1.2}
\setlength{\tabcolsep}{4pt}
\small

\begin{tabular}{|p{3.2cm}|p{4.2cm}|}
\hline
\textbf{Method} & \textbf{Key Limitation} \\
\hline

JSCC-based SemCom~\cite{task_oriented_comm1, task_oriented_comm2} &
No token-level representation; lacks multi-user and scheduling capability. \\
\hline

Knowledge-driven semantic systems~\cite{IoTkadam2025, kadam2023knowledge, kadam2024knowledge, kadam2024semantic} &
Depend heavily on prior knowledge; not scalable to heterogeneous cross-modal data. \\
\hline

Transformer / LSTM models~\cite{dosovitskiy2021image,hochreiter1997long} &
Designed for representation learning, not communication or resource allocation. \\
\hline

Token pruning / ATS methods~\cite{devoto2024adaptive, kim2022efficient} &
Single-user focus; ignore wireless constraints and inter-user interference. \\
\hline

Semantic token selection works~\cite{yang2022semantic, qin2022semantic, feng2023semantic, wang2022semantic, zhang2023adaptive} &
Do not model semantic interference or multi-user scheduling. \\
\hline

OMA / NOMA systems~\cite{noma_survey} &
Operate at signal level; ignore semantic structure and token relationships. \\
\hline

Semantic NOMA systems~\cite{li2023non} &
Though operate at semantic level for multi-user communications; but ignore semantic interference and token relationships. \\
\hline

\textbf{Proposed ATS-ToDMA} &
Jointly models token selection, semantic interference, and multi-user scheduling in a unified framework. \\
\hline

\end{tabular}
\end{table}
\section{Related Work}
Semantic communication has recently emerged as a promising paradigm for next-generation wireless systems, aiming to transmit task-relevant information instead of raw data. Early works in task-oriented communication demonstrated that optimizing end-task performance, rather than bit-level accuracy, can significantly improve communication efficiency under bandwidth and channel constraints. In particular, deep learning-based joint source-channel coding (JSCC) approaches have shown strong performance in multimedia transmission by directly mapping source data into channel inputs, enabling robust semantic delivery over wireless links \cite{task_oriented_comm1, task_oriented_comm2}.

More recently, fundamental theoretical frameworks for SemCom have been developed, including semantic entropy, semantic rate-distortion theory, and goal-oriented communication principles \cite{semantic_comm_overview1, semantic_comm_overview2, xin2024semantic, chaccour2025less}. These works formalize the shift from symbol-level fidelity to semantic-level task fidelity, motivating new design principles and performance metrics tailored to downstream task success rather than reconstruction accuracy.

Building upon these foundations, recent studies have extended SemCom to multi-user and resource-constrained wireless systems. For instance, joint communication and computation optimization has been investigated in probabilistic SemCom frameworks, highlighting the importance of semantic-aware resource allocation \cite{zhao2024joint}. Additionally, recent works have studied robustness against channel impairments and semantic distortion, emphasizing the need for interference-aware semantic system design \cite{peng2025robust}. However, these approaches largely rely on conventional signal-level abstractions and do not explicitly model semantic-level interactions in multi-user environments.

SemCom has gained significant attention in recent years. Prior works such as \cite{IoTkadam2025, kadam2023knowledge, kadam2024knowledge, kadam2024semantic} have explored knowledge-driven SemCom systems, where contextual and domain knowledge is leveraged to enhance transmission efficiency and task performance.

In parallel, deep sequence modeling architectures have demonstrated strong representation capabilities. Transformers have shown excellent performance in capturing long-range dependencies through self-attention mechanisms \cite{dosovitskiy2021image}, while recurrent architectures such as LSTM remain widely used for efficient temporal modeling in sequential data processing \cite{hochreiter1997long}.

To further improve efficiency, adaptive token selection (ATS) and token pruning techniques have been widely studied to reduce computational complexity while preserving task-relevant information \cite{devoto2024adaptive, kim2022efficient}. These methods selectively retain informative tokens based on attention or importance scores.

In semantic and edge intelligence applications, several studies \cite{yang2022semantic, qin2022semantic, feng2023semantic, wang2022semantic, zhang2023adaptive} highlight the importance of token-level selection mechanisms for improving communication efficiency, reducing redundancy, and enabling scalable semantic transmission over wireless networks.

Multiple access techniques have also been explored in SemCom systems. Classical orthogonal multiple access (OMA) and non-orthogonal multiple access (NOMA) schemes have been adapted to improve spectral efficiency in semantic transmission scenarios \cite{noma_survey}. Nevertheless, these methods fail to exploit semantic relationships among transmitted tokens, which can lead to performance degradation in multi-user cross-modal settings. In particular, they do not account for semantic interference arising from similarity in semantic representations.

More recently, transformer-based architectures have been investigated for scheduling and resource allocation due to their ability to capture long-range dependencies and complex interactions. However, their application in SemCom remains limited, particularly in multi-user cross-modal scenarios.

In contrast to existing works, this paper proposes a unified ATS-ToDMA framework for multi-user cross-modal SemCom system. The proposed framework jointly integrates adaptive token selection, semantic-aware scheduling, and wireless resource allocation. By introducing a transformer-based scheduler and explicitly modeling semantic interference, the proposed approach enables efficient token-level multiple access while accounting for both intra-modal and cross-modal semantic dependencies, thereby providing a comprehensive solution for next-generation SemCom systems.

The comparison in Table~\ref{tab:comparison_gain} highlights several important observations. First, existing SemCom frameworks, including JSCC-based and knowledge-driven approaches, primarily focus on improving end-task performance but do not incorporate token-level transmission or multi-user resource coordination. Second, while adaptive token selection and pruning methods improve computational efficiency, they are largely designed for isolated single-user settings and do not consider wireless resource constraints or inter-user interference. Third, classical multiple access schemes such as OMA and NOMA extend connectivity to multi-user scenarios but operate at the signal level and fail to exploit semantic structure, thereby overlooking semantic relationships among transmitted representations.

More importantly, none of the existing approaches jointly address token selection, multi-user scheduling, and semantic interference modeling within a unified framework. In contrast, the proposed ATS-ToDMA framework integrates all three components by treating semantic tokens as the fundamental transmission units and explicitly modeling their interdependencies through a transformer-based scheduler. This enables joint optimization of token importance, resource allocation, and semantic interference mitigation, which is essential for achieving reliable and efficient cross-modal SemCom in multi-user environments.

Although substantial progress has been made in semantic communication, token pruning, and multi-user access techniques, jointly optimizing semantic token selection, semantic interference management, and wireless resource allocation remains largely unexplored. The difficulty arises because token importance, semantic similarity, and channel conditions are strongly coupled, making conventional resource-allocation techniques unsuitable for semantic communication systems.
\begin{table*}
\centering
\caption{Comparison of existing methods with the proposed ATS-ToDMA framework including capability and performance characteristics.}
\label{tab:comparison_gain}
\renewcommand{\arraystretch}{1.15}
\setlength{\tabcolsep}{4pt}
\small

\begin{tabular}{|l|c|c|c|c|c|c|}
\hline
\textbf{Method} &
\textbf{SemCom} &
\textbf{ATS} &
\textbf{Multi-U} &
\textbf{Sem. Int.} &
\textbf{Token Sch.} &
\textbf{Gain} \\
\hline

JSCC-based \cite{task_oriented_comm1, task_oriented_comm2}
& \checkmark & $\times$ & $\times$ & $\times$ & $\times$ & Low \\
\hline

Knowledge-driven~\cite{IoTkadam2025, kadam2023knowledge, kadam2024knowledge, kadam2024semantic}
& \checkmark & $\times$ & $\times$ & $\times$ & $\times$ & Medium \\
\hline

Transformer Models \cite{dosovitskiy2021image}
& $\times$ & $\times$ & $\times$ & $\times$ & $\times$ & N/A \\
\hline

LSTM Models \cite{hochreiter1997long}
& $\times$ & $\times$ & $\times$ & $\times$ & $\times$ & N/A \\
\hline

Token Pruning \cite{devoto2024adaptive, kim2022efficient}
& $\times$ & \checkmark & $\times$ & $\times$ & $\times$ & Low-Med \\
\hline

Semantic Token Works \cite{yang2022semantic, qin2022semantic, feng2023semantic, wang2022semantic, zhang2023adaptive}
& \checkmark & \checkmark & $\times$ & $\times$ & $\times$ & Medium \\
\hline

OMA/NOMA~\cite{noma_survey}
& \checkmark & $\times$ & \checkmark & $\times$ & $\times$ & Medium \\
\hline

Semantic NOMA ~\cite{li2023non}
& \checkmark & $\times$ & \checkmark & $\times$ & $\times$ & Medium \\
\hline

\textbf{Proposed ATS-ToDMA}
& \checkmark & \checkmark & \checkmark & \checkmark & \checkmark & \textbf{High} \\
\hline

\end{tabular}
\end{table*}
\section{System Model and Problem Formulation}\label{Sec:SysModel_ProbForm}
\begin{figure}
\centering
\includegraphics[width=0.48\textwidth]{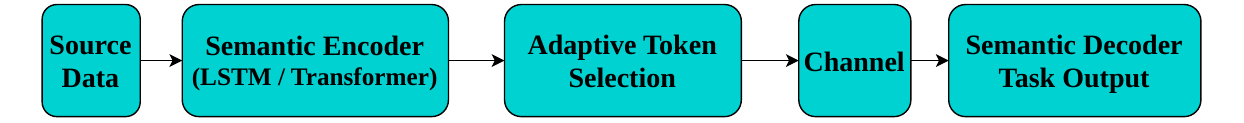}
\caption{SemCom architecture with adaptive token selection.}
\label{fig:system}
\end{figure}
In this section, first, we provide a brief overview of the proposed system model and later in Section~\ref{Sec:Prob_Formulation}, we present our problem formulation. 

\subsection{System Model} \label{Sec:SysModel}
The overall operation of the proposed SemCom architecture with adaptive token selection system, shown in Fig.~\ref{fig:system}, is described as follows. First, the source generates the input data denoted by $x$, which represents the raw information to be transmitted. This input is then processed by a semantic encoder that extracts meaningful and task-relevant features from the data. The encoding process is expressed as $z = f(x)$,
where $f(\cdot)$ denotes the encoding function and $z$ represents the latent semantic representation of the input.

Following feature extraction, an adaptive token selection (ATS) mechanism is employed to improve transmission efficiency (see Algorithm~\ref{Algo:TokenSelection}). Specifically, the ATS module selects a subset of the most informative tokens from the full set of encoded tokens. Let $N$ denote the total number of tokens and $K \subset N$ represent the selected subset. This selection process ensures that only the most significant tokens, with respect to the underlying task, are retained for transmission, thereby reducing redundancy and communication overhead.

\begin{algorithm}
\caption{Adaptive Token Selection (ATS)}
\begin{algorithmic}[1]
\STATE Input tokens $T=\{t_1,...,t_N\}$
\STATE Compute importance scores $s_i = q(t_i)$
\FOR{each token $t_i$}
\IF{$s_i > \tau_{\text{ATS}}$}
\STATE Retain $t_i$
\ENDIF
\ENDFOR
\STATE Output selected tokens $K$
\end{algorithmic}
\label{Algo:TokenSelection}
\end{algorithm}

The selected tokens are then transmitted over the communication channel, which may introduce noise or distortion depending on channel conditions. Despite these impairments, the system is designed to preserve the most critical semantic information through the selective transmission process.
At the receiver side, a decoder processes the received tokens to reconstruct the task-relevant information. The decoder leverages the semantic structure embedded in the transmitted tokens to recover an accurate representation of the original input or its relevant features. This end-to-end process enables efficient and robust communication by focusing on semantic fidelity rather than exact signal reconstruction.

Next, this ATS mechanism is used in a multi-user cross-modal SemCom system as shown in Fig.~\ref{fig:ATS_ToDMA_Framework}. In the considered SemCom framework, each input is mapped to a sequence of semantic tokens through a modality-specific semantic encoder (e.g., LLM~\cite{yao2024survey}, ViT~\cite{yuan2021tokens}, or audio). One of the main components in every token is an embedding vector $\mathbf{e}_i \in \mathbb{R}^d$, where $d$ denotes the embedding dimension. Without loss of generality, embeddings are normalized such that $\|\mathbf{e}_i\| = 1$, which simplifies similarity analysis and stabilizes training.
\begin{figure}
\centering
\includegraphics[width=0.48\textwidth]{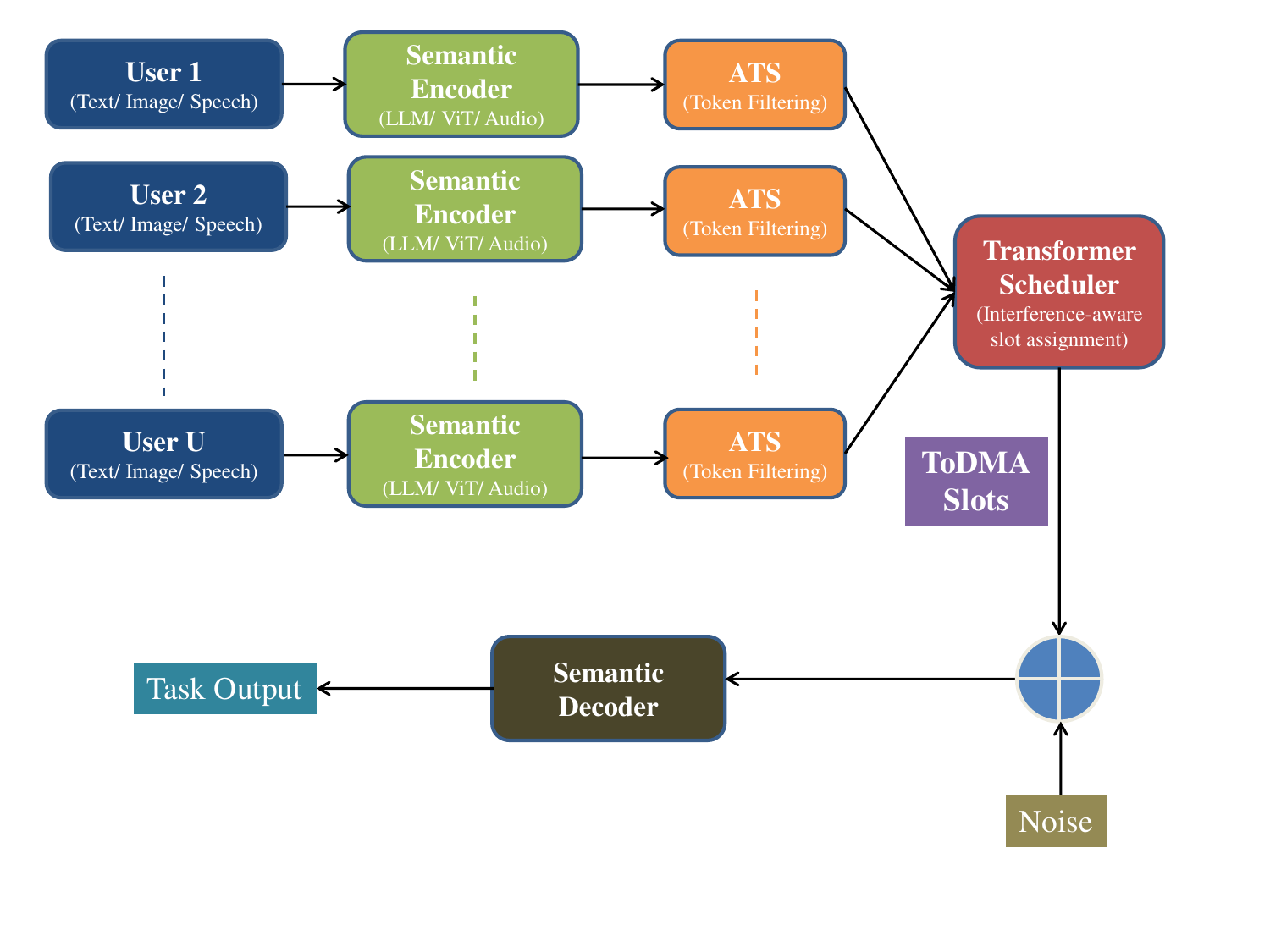}
\caption{Architecture of the proposed ATS-ToDMA framework for multi-user cross-modal SemComs.}
\label{fig:ATS_ToDMA_Framework}
\end{figure}

\subsubsection{Token Generation and Adaptive Selection}
Let us consider $U$ users and each user $u$ generates modality-specific semantic tokens through a semantic encoder. For modality $m \in \{\text{text, image, speech}\}$, the token set is given by:
\begin{equation}
T_u^{(m)} = \{ t_1^{(u,m)}, \dots, t_{K_u^{(m)}}^{(u,m)} \}.
\end{equation}

Since not all tokens contribute equally to semantic understanding, transmitting all tokens leads to inefficiency and unnecessary interference. To address this, we adopt an Adaptive Token Selection (ATS) mechanism (see Fig.~\ref{fig:system}) that filters tokens based on their semantic importance:
\begin{equation}
t_i^{(u,m)} \in T_u^{(m)} \quad \text{if} \quad s_i^{(u,m)} > \tau_{\text{ATS}},
\end{equation}
where $s_i^{(u,m)}$ denotes the importance score and $\tau_{\text{ATS}}$ is a threshold.
This step reduces redundancy and ensures that only semantically meaningful tokens are forwarded for transmission.

\subsubsection{Transformer-Based Semantic Scheduling}
Conventional greedy scheduling schemes allocate resources based solely on local importance, ignoring global semantic dependencies across users and modalities. This often leads to suboptimal decisions under interference constraints.
To overcome this limitation, we propose a transformer-based scheduler that jointly models token importance and semantic relationships.

Each token is represented as:
$\mathbf{x}_i = [\mathbf{e}_i, s_i]$,
where $\mathbf{e}_i$ is the embedding and $s_i$ is the importance score. The transformer encoder produces contextualized representations:
\begin{equation}
\mathbf{Z} = \text{Transformer}(\mathbf{X}).
\end{equation}

The self-attention mechanism enables each token to capture dependencies with all other tokens:
\begin{equation}
\text{Attention}(\mathbf{Q},\mathbf{K},\mathbf{V}) =
\text{softmax}\left(\frac{\mathbf{Q}\mathbf{K}^\top}{\sqrt{d}}\right)\mathbf{V}.
\end{equation}
Based on these representations, the scheduler predicts the probability of assigning token $i$ to slot $k$:
\begin{equation}
p_{ik} = \Pr(\text{token } i \rightarrow \text{slot } k).
\end{equation}
This formulation enables globally optimized scheduling decisions that account for semantic interactions.
At deployment, strict constraint satisfaction is required. Tokens are first assigned via:
\begin{equation}
\text{slot}(i) = \arg\max_k p_{ik},
\end{equation}
followed by pruning to enforce slot capacity constraints and interference thresholds (see Section~\ref{Sec:ConstraintSchedule}).

\subsubsection{Joint Multi-Modal Semantic Decoder}
At the receiver, tokens from different modalities may arrive within the same slot. To exploit complementary information, we propose a joint multi-modal decoder.

\paragraph{Intra-Modal Refinement}
Each modality $m \in \{\text{text, image, speech}\}$ is first refined independently:
\begin{equation}
\tilde{\mathbf{Z}}_m = \text{Attention}(\mathbf{Z}_m, \mathbf{Z}_m, \mathbf{Z}_m).
\end{equation}

\paragraph{Cross-Modal Interaction}

Cross-attention enables information exchange across modalities, improving semantic reconstruction under ambiguity.

\paragraph{Fusion}

The final representation aggregates all modalities:
\begin{equation}
\mathbf{Z}_{\text{joint}} =
\text{Concat}(\mathbf{Z}_T^{\text{joint}}, \mathbf{Z}_I^{\text{joint}}, \mathbf{Z}_S^{\text{joint}}).
\end{equation}

\subsubsection{Channel-aware Semantic Decoding}
To enhance robustness under dynamic wireless channels, we incorporate channel state information (CSI) and signal-to-noise ratio (SNR) into the semantic decoding process. In contrast to conventional semantic decoders that operate solely in the representation space, the proposed framework introduces a channel-aware modulation mechanism that adaptively adjusts token-level semantic importance according to instantaneous channel conditions. Specifically, CSI and SNR are embedded into a learnable gating function that modulates token representations prior to cross-modal fusion. This enables the decoder to suppress semantically unreliable tokens under poor channel conditions while preserving informative semantic features when channel quality is high, thereby improving overall semantic reliability.

Let:
\begin{itemize}
    \item $h_i$: CSI feature for token stream $i$
    \item $\gamma_i$: instantaneous SNR estimate
    \item $Z_i$: token representation
\end{itemize}

We construct a channel-aware embedding gate:
\begin{equation}
\mathbf{g}_i = \sigma\left(\mathbf{W}_h \mathbf{h}_i + \mathbf{W}_\gamma \gamma_i \right),
\end{equation}
where:
\begin{itemize}
    \item $\mathbf{W}_h$ and $\mathbf{W}_\gamma$ are learnable parameter matrices that project CSI features and instantaneous SNR estimates into the semantic gating space.
    \item $\sigma(\cdot)$ is the sigmoid function
    \item $\mathbf{g}_i \in (0,1)^d$ acts as a semantic reliability mask
\end{itemize}

\begin{equation}
\tilde{\mathbf{Z}}_i = \mathbf{g}_i \odot \mathbf{Z}_i,
\end{equation}
This means bad channel has suppressed semantic features and good channel has preserved semantic content. Next, we integrate CSI/SNR into attention as follows:
\begin{equation}
\text{Attention} =
\text{softmax}\left(
\frac{\mathbf{Q}\mathbf{K}^\top}{\sqrt{d}} + \beta \mathbf{A}_{\text{ch}}
\right)\mathbf{V}
\end{equation}
where:
\begin{equation}
A_{\text{ch}} = f(\text{CSI}, \text{SNR})
\label{eq:Attention_Ch}
\end{equation}

\begin{table}[t]
\caption{Major Notations Used in the ATS-ToDMA Framework}
\label{tab:notations}
\centering
\begin{tabular}{|c|p{0.38\textwidth}|}
\hline
\textbf{Symbol} & \textbf{Description} \\
\hline

$u$ & User index. \\
\hline

$t_i$ & Semantic token corresponding to the $i$-th semantic unit. \\
\hline

$e_i$ & Embedding vector associated with token $t_i$. \\
\hline

$d$ & Embedding dimension. \\
\hline

$s_i$ & Semantic importance score of token $t_i$. \\
\hline

$x_i$ & Binary token selection variable. \\
\hline


$P_i$ & Transmit power allocated to token $t_i$. \\
\hline

$\gamma_i$ & Instantaneous SNR associated with token $t_i$. \\
\hline

$g_i$ & Channel-aware semantic protection factor for token $t_i$. \\
\hline

$\xi_{ij}$ & Semantic similarity between tokens $t_i$ and $t_j$. \\
\hline

$\delta$ & Maximum admissible semantic similarity. \\
\hline

$\gamma$ & Semantic similarity threshold used in ATS. \\
\hline

$\alpha_{ij}$ & Semantic interference coefficient between tokens $i$ and $j$. \\
\hline

$\alpha_{\mathrm{intra}}$ & Average intra-modal semantic interference coefficient. \\
\hline

$\alpha_{\mathrm{cross}}$ & Average cross-modal semantic interference coefficient. \\
\hline

$I_{ij}$ & Pairwise semantic interference between tokens $i$ and $j$. \\
\hline

$I_{\mathrm{total}}$ & Aggregate semantic interference within a ToDMA slot. \\
\hline

$\hat I_k$ & Expected semantic interference in slot $k$ during scheduler optimization. \\
\hline

$N_0$ & Thermal noise power. \\
\hline

$\mathrm{SSINR}_i$ & Semantic Signal-to-Interference-plus-Noise Ratio of token $i$. \\
\hline

$\Gamma$ & Minimum target SSINR threshold. \\
\hline

$M$ & Number of simultaneously scheduled tokens within a ToDMA slot. \\
\hline

$M_{\max}$ & Maximum feasible token occupancy satisfying the SSINR constraint. \\
\hline

$R_s$ & Semantic throughput. \\
\hline

\end{tabular}
\end{table}

\subsection{Problem Formulation}\label{Sec:Prob_Formulation}

When the tokens are used for transmission, they can interfere in two situations: The first situation is if they are semantically similar, and the second situation is if they occupy the same ToDMA slot. 

\subsubsection{Semantic Similarity}
To define the semantic similarity, we need to first explore the cosine similarity. Let $\xi_{ij}$ be the cosine similarity between tokens $t_i$ and $t_j$ and it is defined as
\begin{equation}
\xi_{ij} = \frac{\langle t_i, t_j \rangle}{\|t_i\|\|t_j\|}.
\end{equation}

We consider a token pair $(t_i,t_j)$ is \textit{semantically similar} if $\xi_{ij} > \overline{\gamma}$. Let 
\begin{equation}
\mathbbm{1}_{ij} =
\begin{cases}
1, & \xi_{ij} > \overline{\gamma}, \\
0, & \text{otherwise}.
\end{cases}
\end{equation}

The total semantically similar token pairs are
\begin{equation}
I_{ss} = \sum_{i \in T_u^{(m)}} \sum_{j \in T_u^{(m)}, ~j \neq i} \mathbbm{1}_{ij}.
\end{equation}

Next, when multiple tokens are transmitted within the same ToDMA slot, interference arises not only from waveform overlap but also from semantic ambiguity. Let $\rho_{ij} = \frac{P_iP_j}{P_{ref}}$. We model the pairwise semantic interference as follows:
\begin{equation}
I_{ij} = \alpha_{ij} \rho_{ij} \xi_{ij}^2 \mathbbm{1}_{ij},
\end{equation}
where 
\begin{equation}
    \alpha_{ij} =
        \begin{cases}
        \alpha_{\text{intra}}, & \text{same modality} \\
        \alpha_{\text{cross}}, & \text{different modality}
        \end{cases}
\end{equation}
and $P_i$ denote the transmission power assigned to the token $i$ and reference interference power $P_{ref}$.\footnote{Unlike electromagnetic interference, semantic interference does not originate from waveform superposition. Instead, it reflects ambiguity introduced by semantically similar representations during downstream task inference.} The squared similarity term emphasizes highly correlated token pairs while suppressing weak semantic interactions, analogous to the role of power in conventional SINR expressions. The inclusion of the semantic distortion coefficient $\alpha_{ij}$ enables a unified interpretation of semantic interference in the signal domain. 
The power term $\rho_{ij}$ captures the joint contribution of the interacting token pair to semantic interference. Since $\rho_{ij}$ has units of power, the semantic coupling coefficient $\alpha_{ij}$ is dimensionless. Consequently, the semantic interference term $I_{ij}$ retains the physical units of power while preserving the pairwise nature of semantic interactions.

In particular, distinguishing between intra-modal and cross-modal coefficients allows the model to capture the reduced effective distortion induced by cross-modal interactions, i.e., $\alpha_{\text{cross}} < \alpha_{\text{intra}}$.\footnote{The assumption $\alpha_{cross}<\alpha_{intra}$ is motivated by the complementary nature of heterogeneous modalities. Tokens originating from the same modality often share similar feature spaces and semantic structures, resulting in higher decoding ambiguity when semantic overlap occurs. In contrast, cross-modal tokens provide diverse contextual cues that can aid semantic disambiguation during fusion. Therefore, cross-modal interactions generally induce lower effective semantic distortion compared with intra-modal interactions.} This leads to less conservative interference bounds and improved slot utilization.

The values of $\alpha_{\text{intra}}$ and $\alpha_{\text{cross}}$ are obtained through the calibration procedure described in Section~\ref{Sec:Calibration}. For a given semantic decoder, intra-modal and cross-modal token pairs are evaluated separately over a validation set. The average semantic distortion is computed as
\begin{equation}
\alpha_{\text{intra}}
=
\mathbb{E}
[\alpha_{ij}\mid m_i=m_j]
\end{equation}

and

\begin{equation}
\alpha_{\text{cross}}
=
\mathbb{E}
[\alpha_{ij}\mid m_i\neq m_j].
\end{equation}

The aggregate interference in a slot containing $M$ tokens is therefore
\begin{equation}
I_{\text{total}} = \sum_{i \neq j} I_{ij} = \sum_{i \neq j} \alpha_{ij} \rho_{ij} \xi_{ij}^2 \mathbbm{1}_{ij}.
\label{eq:I_Total}
\end{equation}

\subsubsection{Cross-Modal Similarity Metric} 
In cross-modal SemCom, tokens originate from different modalities, namely text (T), image (I), and speech (S). Let the corresponding token embeddings be defined as:
\begin{equation}
\mathbf{\tau}_i \in \mathbb{R}^d, \quad 
\mathbf{\phi}_j \in \mathbb{R}^d, \quad 
\mathbf{\psi}_k \in \mathbb{R}^d
\end{equation}

The expected cross-modal similarity is defined as:
\begin{equation}
\bar{s}_{\text{cross}} = 
\mathbb{E}[\xi(\mathbf{\tau}_i, \mathbf{\phi}_j)] 
= \mathbb{E}[\xi(\mathbf{\tau}_i, \mathbf{\psi}_k)] 
= \mathbb{E}[\xi(\mathbf{\phi}_j, \mathbf{\psi}_k)]
\end{equation}

Consider a ToDMA slot containing $K_T$, $K_I$, and $K_S$ tokens from text, image, and speech modalities, respectively. The total semantic interference is:

\begin{align}
I_s& = 
\underbrace{
\sum_{i \neq j}^{K_T} \xi(\tau_i, \tau_j)
+ \sum_{i \neq j}^{K_I} \xi(\phi_i, \phi_j)
+ \sum_{i \neq j}^{K_S} \xi(\psi_i, \psi_j)
}_{\text{intra-modal interference}} \nonumber \\
& +
\underbrace{
\sum_{i=1}^{K_T} \sum_{j=1}^{K_I} \xi(\tau_i, \phi_j)
+ \sum_{i=1}^{K_T} \sum_{k=1}^{K_S} \xi(\tau_i, \psi_j)
+ \sum_{i=1}^{K_I} \sum_{k=1}^{K_S} \xi(\phi_i, \psi_j)
}_{\text{cross-modal interference}}
\end{align}

\subsubsection{Semantic Reliability Constraint}
To ensure reliable semantic decoding, we introduce a basic semantic SINR metric:
\begin{equation}
\overline{\text{SSINR}}(i) =
\frac{P_i}{
\sum_{j \neq i} \alpha_{ij} \rho_{ij} \xi_{ij}^2  + N_0}
\label{eq:SSINR}
\end{equation}
To ensure consistency between semantic interference and physical-layer noise, we introduce a semantic-to-signal distortion mapping. Specifically, semantic similarity between tokens is interpreted as a source of decoding ambiguity, which manifests as an effective distortion in the received signal. This is captured using a scaling factor $\alpha_{ij}$, which maps semantic similarity into an equivalent interference power. Furthermore, to distinguish between intra-modal and cross-modal effects, we define modality-dependent coefficients such that $\alpha_{\text{cross}} < \alpha_{\text{intra}}$, reflecting the reduced ambiguity induced by cross-modal interactions.

\subsubsection{Calibration of Semantic Interference Coefficients}\label{Sec:Calibration}
The semantic distortion coefficient $\alpha_{ij}$ provides a mapping between semantic ambiguity and an equivalent physical-layer interference representation. Specifically, $\alpha_{ij}$ quantifies how semantic overlap between two tokens degrades end-task performance.
To obtain a measurable value of $\alpha_{ij}$, we employ an offline calibration procedure. Consider a reference semantic decoder operating under nominal channel conditions. Let $A_0$ denote the task accuracy achieved when token $t_i$ is decoded without semantic interference.
Next, an interfering token $t_j$ with similarity $\xi_{ij}$ is introduced, and the resulting task accuracy becomes $A_{ij}$.
The semantic distortion induced by token overlap is defined as
$D_{ij} = A_0-A_{ij}$.

To relate semantic distortion to an equivalent interference power, we normalize this degradation with respect to the accuracy loss caused by a known physical interference level $P_{\mathrm{ref}}$:
\begin{equation}
\alpha_{ij}
=
\frac{D_{ij}}
{D_{\mathrm{ref}}},
\end{equation}
where $D_{\mathrm{ref}}$ denotes the accuracy degradation observed under the reference interference power $P_{\mathrm{ref}}$. Consequently, $\alpha_{ij}$  acts as a semantic-to-physical conversion factor, enabling semantic interference and thermal noise to be represented in a common SSINR framework.

\subsubsection{Semantic Protection}
To account for the fact that each token experiences its own
instantaneous channel quality, from~\eqref{eq:Attention_Ch} the channel-aware protection term is defined as
\begin{equation}
A_{ch}(i,j)
\!=\!
\sqrt{\log(1+\gamma_i)\log(1+\gamma_j)}
\exp(-\|h_i-h_j\|^2),
\end{equation}
where $\gamma_i$ and $\gamma_j$ denote the SNRs associated
with tokens $i$ and $j$, respectively.
This gives a 3-layer semantic protection mechanism:
\begin{enumerate}
    \item \textbf{Semantic layer:} Semantic similarity $\xi_{ij}$
    \item \textbf{Modality layer:} Semantic distortion coefficients  $\alpha_{\text{intra}}$ and $\alpha_{\text{cross}}$
    \item \textbf{Channel layer:} The combination of CSI and SNR gating.
\end{enumerate}

So, our $\overline{\text{SSINR}}$ becomes channel-aware semantic SINR:
\begin{equation}
\text{SSINR}_i =
\frac{P_i \|\mathbf{g}_i\|^2}
{\sum_{j \neq i} \alpha_{ij} \rho_{ij} \xi_{ij}^2 \|\mathbf{g}_j\|^2 + N_0}
\label{eq:SSINR_i}
\end{equation}
Unlike conventional wireless interference, semantic interference originates from ambiguity in the representation space. The calibration procedure described in Section~\ref{Sec:Calibration} maps this ambiguity into an equivalent interference power through the coefficient $\alpha_{ij}$. Consequently, both semantic interference and thermal noise contribute to task degradation and can be jointly represented within the SSINR expression. This interpretation allows semantic reliability analysis using tools analogous to conventional SINR-based communication theory.

\subsubsection{Semantic Throughput}
Let $\mathcal{T}$ be the total token set across modalities, i.e.:
\begin{equation}
\mathcal{T} = \mathcal{T}_{\text{text}} \cup \mathcal{T}_{\text{image}} \cup \mathcal{T}_{\text{speech}}.
\end{equation}
The semantic throughput is defined as
\begin{equation}
R_s = \frac{|\mathcal{T}|}{T_{\text{slots}}}.
\end{equation}
Let $x_i \in \{0,1\}$ denote token selection and $p_{ik}$ denote slot assignment probabilities. Then $\sum_k p_{ik} = x_i, \ \forall i$ and $\sum_i x_i = R_s$. 
\subsubsection{Optimization Problem}
We jointly optimize token selection, slot assignment, and power allocation to maximize semantic throughput under semantic interference and channel constraints.
Semantic throughput is selected as the optimization objective because it jointly captures semantic importance, scheduling decisions, and transmission reliability through the SSINR metric. Consequently, maximizing semantic throughput encourages efficient utilization of both semantic and wireless resources.

\begin{subequations}
\begin{align}
\max_{\{x_i, p_{ik}, P_i\}} \quad 
& \frac{\sum_i x_i}{T_{\text{slots}}} \\
\text{s.t.} \quad 
& \sum_k p_{ik} = x_i, \ \forall i, \\
& \sum_i p_{ik} \leq M_{\max}, \ \forall k, \\
& \mathbb{E}[I_s] \leq \eta, \\
& I_{\text{total}} \leq I_{\max}, \\
& \text{SSINR}_i \geq \Gamma, \ \forall i, \\
& x_i \in \{0,1\}, \quad 0 \leq p_{ik} \leq 1.
\end{align}
\label{eq:Opt}
\end{subequations}
Constraint~(\ref{eq:Opt}b) guarantees consistency between token selection and slot assignment by ensuring that a token is assigned to a slot only if it is selected for transmission. Constraint~(\ref{eq:Opt}c) limits the number of simultaneously transmitted tokens per slot to avoid excessive semantic congestion. Constraint~(\ref{eq:Opt}d) restricts the expected semantic interference arising from cross-modal interactions, whereas~(\ref{eq:Opt}e) imposes a hard upper bound on the instantaneous aggregate interference within each ToDMA slot. Constraint~(\ref{eq:Opt}f) ensures reliable semantic decoding by maintaining the SSINR of each transmitted token above the target threshold $\Gamma$. Finally,~(\ref{eq:Opt}g) specifies the binary nature of token selection decisions and the probabilistic relaxation of slot assignment variables.

This formulation unifies semantic-level token selection, transformer-based scheduling, and channel-aware power allocation into a single optimization framework, enabling holistic control of semantic throughput and reliability.

Due to combinatorial and non-convex constraints, the problem is solved using a hybrid approach combining learning-based scheduling with analytical constraint enforcement.

\section{Proposed ATS-ToDMA Framework for Semantic Communications}
We propose the Adaptive Token Selection and Token-Domain Multiple Access (ATS-ToDMA) framework to maximize semantic throughput while upper-bounding the  semantic and cross-modal interferences and also ensuring a reliable semantic decoding.

\subsubsection{Expected Semantic Interference with Modality-Aware Distortion}

To account for the different impact of intra-modal and cross-modal semantic interactions, we introduce modality-dependent semantic distortion coefficients $\alpha_{\text{intra}}$ and $\alpha_{\text{cross}}$, where typically $\alpha_{\text{cross}} < \alpha_{\text{intra}}$.

Assuming average intra-modal similarity $\bar{s}_{\text{intra}}$ and cross-modal similarity $\bar{s}_{\text{cross}}$, the expected semantic interference can be expressed as:
\begin{align}
\mathbb{E}&[I_s] = (K_T K_I 
+ K_T K_S 
+ K_I K_S) \alpha_{\text{cross}}\bar{s}_{\text{cross}} + \nonumber \\
&\left(\frac{K_T(K_T-1)}{2}  
+ \frac{K_I(K_I-1)}{2}  
+ \frac{K_S(K_S-1)}{2}\right) \alpha_{\text{intra}}\bar{s}_{\text{intra}}  
\end{align}

To satisfy an interference constraint $\mathbb{E}[I_s] \leq \eta$, we obtain the following bound:

\begin{align}
\mathbb{E}[I_s] = ~&\frac{\alpha_{\text{intra}}\bar{s}_{\text{intra}}}{2}
\sum_{m \in \{T,I,S\}} K_m (K_m - 1) \nonumber \\
&+ \alpha_{\text{cross}}\bar{s}_{\text{cross}} (K_T K_I + K_T K_S + K_I K_S)
\leq \eta
\end{align}
This expression provides a closed-form guideline for allocating tokens per modality within each slot.

\begin{theorem}
Assume that the pairwise semantic similarity satisfies $\overline{\gamma} \leq \xi_{ij} \leq \delta$ for all token pairs, and that the semantic distortion coefficients satisfy $0 \leq \alpha_{ij} \leq \alpha_{\max}$. Then, for $M$ tokens transmitted in the same slot with equal power allocation $P$, the total semantic interference is upper bounded as
\begin{equation}
I_{\text{total}} \leq I_{\max} = \alpha_{\max} P \delta^2 M(M-1).
\end{equation}
\end{theorem}
\begin{proof}
The proof is given in Appendix~\ref{Apdx:Thm1}.
\end{proof}
Theorem 1 shows that aggregate semantic interference increases approximately quadratically with token occupancy. Consequently, excessive slot loading can significantly degrade semantic reliability, motivating interference-aware token scheduling in ATS-ToDMA.

\begin{theorem}
Assume that $\overline{\gamma} \leq \xi_{ij} \leq \delta$ and $0 \leq \alpha_{ij} \leq \alpha_{\max}$ for all token pairs, and that the semantic SINR satisfies $\text{SSINR}_i \geq \Gamma$ for all $i$. Then, the number of tokens transmitted in a slot with equal power allocation $P$ is upper bounded by
\begin{equation}
M_{\max} = 1 + \frac{P \|\mathbf{g}_i\|^2 - \Gamma N_0}{\Gamma P \alpha_{\max} \delta^2 d}.
\end{equation}
\end{theorem}

\begin{proof}
The proof is given in Appendix~\ref{Apdx:Thm2}.
\end{proof}
Theorem 2 establishes a direct relationship between channel quality, semantic interference, and feasible token occupancy. Users experiencing favorable channel conditions can support a larger number of simultaneously scheduled semantic tokens, whereas poor channel conditions require more conservative slot allocations to maintain the target SSINR level.

\begin{remark}
The bound in Theorem 2 explicitly depends on the channel-aware semantic protection factor $\|g_i\|^2$, which is a function of the token-specific CSI and SNR. Consequently, users experiencing favorable channel conditions admit larger feasible token occupancies, whereas users under poor channel conditions require more conservative slot allocations to satisfy the SSINR requirement. Therefore, the proposed ATS-ToDMA framework naturally adapts token scheduling and power allocation according to heterogeneous wireless conditions.
\end{remark}

\subsubsection{Constraint-Aware Scheduling Mechanism} \label{Sec:ConstraintSchedule}
While the transformer-based scheduler captures global semantic relationships, it does not inherently guarantee satisfaction of the analytical interference and reliability constraints derived earlier. To address this limitation, we integrate a hybrid constraint-aware scheduling mechanism that operates during both training and inference.

\paragraph{Training Phase (Soft Constraint Enforcement)}
We define the expected number of tokens assigned to slot $k$ as:
\begin{equation}
\hat{M}_k = \sum_i p_{ik},
\end{equation}
and the expected semantic interference in slot $k$ as:
\begin{equation}
\hat{I}_k = \sum_{j \neq i}^M \sum_{i=1}^M p_{ik} p_{jk} \alpha_{ij}\rho_{ij} \xi_{ij}^2 \mathbbm{1}_{ij}.
\label{eq:hat_Ik}
\end{equation}
It is important to distinguish the aggregate semantic interference defined in~\eqref{eq:I_Total} from the expected interference introduced here. Specifically, $I_{total}$ represents the actual interference observed after slot assignments are fixed, whereas $\hat I_k$ denotes the expected interference during training under probabilistic slot assignments generated by the transformer scheduler. The latter provides a differentiable surrogate that enables gradient-based optimization while approximating the deployment-stage interference behavior.

The scheduler is trained using a constraint-aware objective:
\begin{align}
\mathcal{L} =
\mathcal{L}_{\text{task}}
&+ \lambda_1 \sum_k \max(0, \hat{M}_k - M_{\max}) \nonumber \\
&+ \lambda_2 \sum_k \max(0, \hat{I}_k - I_{\max}) \nonumber \\
&+ \lambda_3 \max(0, \mathbb{E}[I_s] - \eta),
\end{align}
where $\lambda_1, \lambda_2,$ and $\lambda_3$ are penalty coefficients.

\paragraph{Inference Phase (Hard Constraint Enforcement)}
After obtaining the slot assignment via
\begin{equation}
\text{slot}(i) = \arg\max_k p_{ik},
\end{equation}
a post-processing step is applied to strictly enforce the system constraints:

\begin{itemize}
\item \textbf{Slot Capacity Constraint:} If the number of tokens in a slot exceeds $M_{\max}$, tokens with the lowest importance scores are removed.
\item \textbf{Interference Constraint:} 
Let $I_c$ be the highest pairwise interference, i.e., 
\begin{equation}
I_c = \arg\max_i \sum_{j \neq i} \alpha_{ij}\rho_{ij} \xi_{ij}^2 = \arg\max_i \sum_{j \neq i} I_{ij}.
\end{equation}
\item If the total interference exceeds $I_{\max}$, tokens contributing $I_c$ are iteratively removed until the constraint is satisfied.
\end{itemize}
This hybrid design ensures that the scheduler remains differentiable and learnable during training while guaranteeing strict adherence to interference and reliability constraints during deployment.

\section{Analysis}
\subsection{Closed-Form Channel-Aware Power Allocation}
Tokens carrying higher semantic importance should receive preferential protection against channel impairments. Therefore, power allocation must balance two competing objectives: preserving highly informative semantic content and exploiting favorable channel conditions to improve resource efficiency. This motivates the proposed semantic-aware power control formulation.
We optimize power allocation to balance semantic importance and channel quality:
\begin{subequations}
\begin{align}
\min_{\{P_i\}} \quad 
&
I_{total}
=
\sum_{i\neq j}^M
\alpha_{ij}\rho_{ij}\xi_{ij}^2. \\
\text{s.t.} \quad 
& \text{SSINR}_i \ge \Gamma.
\end{align}
\label{eq:power_opt}
\end{subequations}
Here, the optimization minimizes the deployment-stage aggregate interference $I_{total}$ defined in~\eqref{eq:I_Total}, rather than the expected training-stage interference $\hat I_k$ in~\eqref{eq:hat_Ik}. The optimal power allocation problem formulated above is inherently non-convex due to the coupled interference terms across tokens. While iterative optimization methods can be applied, they incur significant computational overhead and are unsuitable for real-time semantic communication systems.
To address this limitation, we derive a low-complexity closed-form approximation that provides key insights into the interaction between semantic similarity, interference, and channel quality.

From~\eqref{eq:SSINR_i}, recall the channel-aware semantic SINR (SSINR) as:
\begin{equation}
\text{SSINR}_i =
\frac{
P_i \|\mathbf{g}_i\|^2
}{
\sum_{j \neq i} \alpha_{ij} \rho_{ij} \xi_{ij}^2 \|\mathbf{g}_j\|^2 + N_0
}
\geq \Gamma.
\end{equation}

To simplify notation, we define the effective coupling coefficient:
\begin{equation}
w_{ij} \triangleq \alpha_{ij} \,\xi_{ij}^2 \,\|\mathbf{g}_j\|^2,
\end{equation}
which captures three key factors:
\begin{itemize}
\item semantic similarity ($\xi_{ij}$),
\item modality-dependent interference ($\alpha_{ij}$),
\item channel reliability ($\|\mathbf{g}_j\|^2$).
\end{itemize}

Using this definition, the SSINR constraint can be rewritten as:
\begin{equation}
P_i \|\mathbf{g}_i\|^2 \geq \Gamma \left( \sum_{j \neq i} w_{ij} \rho_{ij} + N_0 \right).
\end{equation}

This expression highlights the fundamental coupling between power allocation and semantic interference. Directly solving the coupled inequalities is computationally expensive. So, in Theorem~\ref{Thm:ApproximatePower} we obtain a tractable closed-form solution by approximating the interference contribution using a first-order expansion. This closed-form expression provides an interpretable structure for power allocation:
\begin{itemize}
\item The first term represents the baseline power required to combat noise.
\item The second term captures additional power needed to overcome semantic interference.
\end{itemize}
This approximation achieves a favorable trade-off between performance and computational efficiency.

\begin{theorem}[Approximate Closed-Form Power Allocation]
The transmit power allocated to token $i$ admits the following approximate closed-form expression:
\begin{equation}
P_i \approx 
\frac{\Gamma N_0}{\|\mathbf{g}_i\|^2}
\left(
1 + \sum_{j \neq i}
\frac{\Gamma\, \alpha_{ij}\, \xi_{ij}^2 \|\mathbf{g}_j\|^2}{\|\mathbf{g}_i\|^2}
\right).
\end{equation}
\label{Thm:ApproximatePower}
\end{theorem}

\begin{proof}
The proof is given in Appendix~\ref{Apdx:Thm3}.
\end{proof}
Theorem 3 reveals that optimal transmit power depends not only on channel quality but also on the semantic coupling among co-scheduled tokens. Tokens experiencing stronger semantic interference require additional protection, whereas semantically isolated tokens can be transmitted using lower power. This observation highlights the importance of jointly considering semantic and physical-layer characteristics during resource allocation.

\begin{remark}
   Theorem 3 indicates that semantically isolated tokens can operate at lower power levels without compromising reliability, thereby improving energy efficiency. 
\end{remark}

The derived expression reveals several important design principles:

\begin{itemize}
\item \textbf{Channel Awareness:} Tokens with better channel conditions ($\|\mathbf{g}_i\|$ large) require less transmit power.
\item \textbf{Semantic Interference:} High semantic similarity ($\xi_{ij}$) increases interference, thereby increasing power requirements.
\item \textbf{Cross-Modal Advantage:} Lower cross-modal interference coefficients ($\alpha_{\text{cross}}$) enable more efficient resource sharing across modalities.
\end{itemize}

These insights validate the importance of jointly modeling semantic, system, and channel effects.

\subsection{Design Insights}
The analytical results developed in Theorems 1--3 provide several important insights into semantic communication system design.
First, semantic interference grows rapidly with token occupancy, indicating that blindly increasing slot utilization may degrade semantic reliability.
Second, feasible token occupancy depends jointly on channel quality and semantic similarity, highlighting the importance of channel-aware semantic scheduling.
Third, power allocation decisions should account for both wireless channel conditions and semantic interference levels. Consequently, semantic and physical-layer resource management cannot be optimized independently.
These observations motivate the integrated ATS-ToDMA framework proposed in this work.

\subsection{Computational Complexity Analysis}
\subsubsection{LSTM Encoder}
The computational complexity of the LSTM encoder is given by
\begin{equation}
C_{\text{LSTM}} = O(N d^2)
\end{equation}
where $N$ is the sequence length and $d$ is the hidden dimension. The sequential nature limits parallelization and increases latency.

\subsubsection{Transformer Encoder}
The transformer encoder has complexity
\begin{equation}
C_{\text{Transformer}} = O(N^2 d)
\end{equation}
due to the self-attention mechanism. While highly parallelizable, it incurs high computation and memory cost.

\subsubsection{Transformer Scheduler}
The transformer scheduler has complexity
\begin{equation}
C_{\text{Scheduler}} = O(K^2 d).
\end{equation}

\subsubsection{Adaptive Token Selection}
ATS selects $K \ll N$ important tokens based on importance score $s_i$:
\begin{equation}
t_i \in K \quad \text{if } s_i > \tau_{\text{ATS}}
\end{equation}
The resulting complexity becomes
\begin{equation}
C_{\text{ATS}} = O(K^2 d)
\end{equation}
The overall computational complexity of ATS-ToDMA can be
expressed as

\[
\mathcal{O}(Kd)
+
\mathcal{O}(K^2 d)
+
\mathcal{O}(K),
\]
where the three terms correspond to adaptive token selection, transformer-based scheduling, and semantic-aware power allocation, respectively.

\section{Simulation Results}
This section evaluates the performance of the proposed
ATS-ToDMA framework and validates the analytical results
established in Theorems~1--3. Unless otherwise stated,
a cross-modal semantic communication network comprising
$U=10$ users is considered, where each user transmits text,
image, and speech tokens. Semantic embeddings are represented
in a 128-dimensional feature space and the semantic similarity
between token pairs is computed using cosine similarity.
The semantic interference coefficients are calibrated according
to Section III-C, with average intra-modal and cross-modal
values of $\alpha_{\mathrm{intra}}=0.8$ and
$\alpha_{\mathrm{cross}}=0.4$, respectively.
The thermal noise power is normalized to $N_0=1$.
All results are averaged over 1000 independent Monte Carlo
realizations.

\subsection{Simulation Setup}
We consider a multi-user semantic communication system supporting text, image, and speech modalities. 
The cosine-similarity statistics are selected to match those observed in representative BERT-, ViT-, and Wav2Vec-based embedding spaces reported in recent semantic communication literature.
The wireless channel is modeled as a Rayleigh fading channel with additive white Gaussian noise (AWGN).
Unless otherwise specified, the ATS similarity threshold is
set to $\gamma=0.5$, the semantic coupling parameter is
$\delta=0.9$, and the target SSINR threshold is $\Gamma=2$.

The proposed ATS-ToDMA framework is compared against the
following benchmark schemes:

\begin{itemize}
\item \textbf{OMA}: Conventional orthogonal multiple access.
\item \textbf{Semantic NOMA}: Non-orthogonal semantic communication without token-aware scheduling.
\item \textbf{Random-TS}: Random token selection and scheduling.
\item \textbf{Greedy ATS}: Token selection based solely on semantic importance.
\item \textbf{Equal-Power Allocation}: Uniform transmit power without semantic-aware optimization.
\end{itemize}

The objective of the simulation study is twofold:
(i) to quantify the performance gains achieved by ATS-ToDMA,
and (ii) to validate the practical usefulness of the analytical
bounds and approximations developed in Theorems~1--3. Performance is evaluated using the following metrics: Semantic Throughput ($R_s$), Semantic Decoding Accuracy, Aggregate Semantic Interference, SSINR, and Average Transmit Power.

\subsection{Semantic Throughput Performance}
Fig.~3(a) illustrates the semantic throughput achieved by
different multiple-access strategies as the number of users
increases.

\begin{equation}
R_s
=
\sum_{i=1}^{N}
x_i s_i
\log_2\left(1+\mathrm{SSINR}_i\right).
\label{eq:semantic_throughput}
\end{equation}
As expected, OMA exhibits the lowest throughput due to
strict orthogonal resource allocation. Semantic NOMA improves
resource utilization through non-orthogonal transmissions,
whereas Random-TS and Greedy ATS benefit from token-level
scheduling. However, these approaches do not explicitly
account for semantic interference.

In contrast, ATS-ToDMA jointly optimizes token selection,
semantic interference management, and channel-aware resource
allocation. Consequently, it consistently achieves the highest
semantic throughput across all network loads. The performance
gap becomes increasingly pronounced as the number of users
grows, demonstrating the scalability of the proposed framework
under dense semantic communication scenarios.

\subsection{Semantic Reliability Performance}
Fig.~3(b) compares the semantic decoding accuracy achieved by
different schemes over a wide range of average SNR values.
ATS-ToDMA achieves superior semantic reliability because
Adaptive Token Selection removes low-value tokens while the
proposed semantic-interference-aware scheduler suppresses
harmful semantic interactions. Consequently, the proposed
framework maintains higher decoding accuracy throughout the
entire SNR range.
The gain is particularly noticeable in the low-to-moderate
SNR regime, where semantic interference management plays
a dominant role in semantic reconstruction quality.

\subsection{Impact of Semantic Similarity Threshold}
Fig.~4 investigates the influence of the semantic similarity
threshold used by the Adaptive Token Selection mechanism.
A small threshold admits a large number of tokens, increasing
semantic redundancy and semantic interference. Conversely,
an excessively large threshold removes potentially useful
semantic information. Therefore, an appropriate threshold
must balance semantic richness and interference suppression.
The results show that ATS-ToDMA achieves its highest semantic
throughput within an intermediate operating region, validating
the effectiveness of the proposed token selection strategy.

\begin{figure}
\centering
\begin{subfigure}{.24\textwidth}
\centering
\begin{adjustbox}{width = 1\columnwidth}
\includegraphics[width=0.99\textwidth]{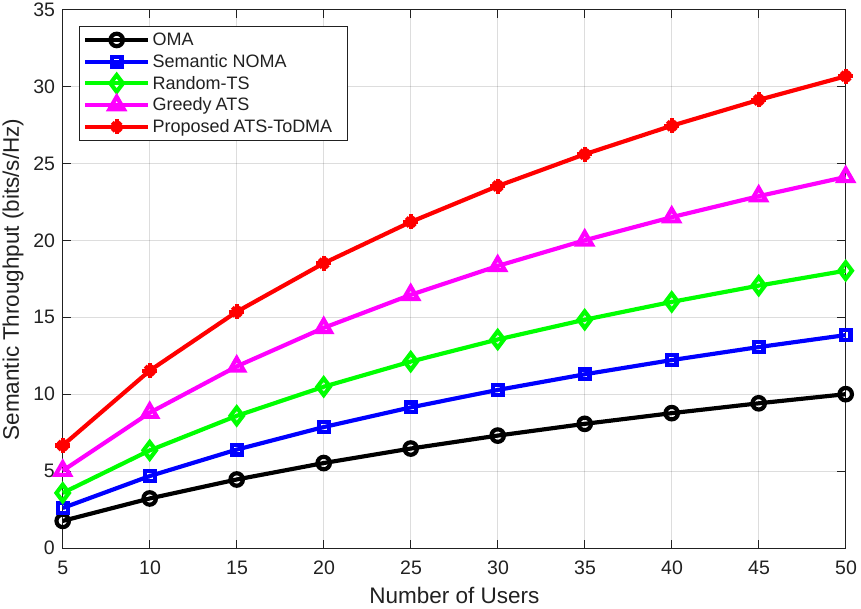}
\end{adjustbox}
\caption{}
\label{fig:QoE_vs_vbar}
\end{subfigure}%
\begin{subfigure}{.24\textwidth}
\centering
\begin{adjustbox}{width = 1\columnwidth}
\includegraphics[width=0.99\textwidth]{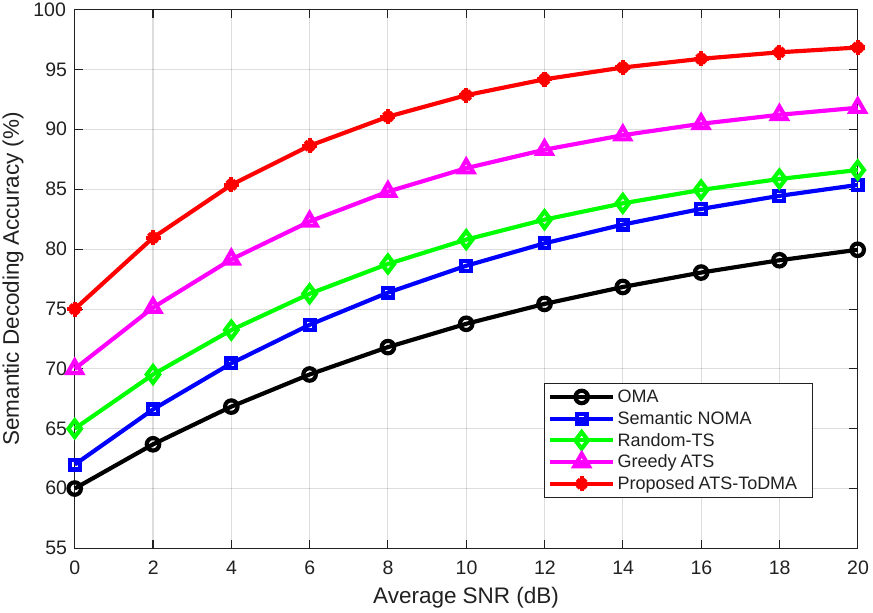}
\end{adjustbox}
\caption{}
\label{fig:QoE_vs_betabar}
\end{subfigure}
\caption{(a) Semantic throughput versus number of users. The proposed ATS-ToDMA framework consistently achieves the highest semantic throughput by jointly optimizing token selection, semantic interference management, and channel-aware resource allocation. (b) Validation of Theorem~2: theoretical and simulated maximum feasible token occupancy as a function of the target SSINR threshold.}
\label{fig:Sem_Decoding}
\end{figure}

\begin{figure}[t]
\centering
\includegraphics[width=0.45\textwidth]{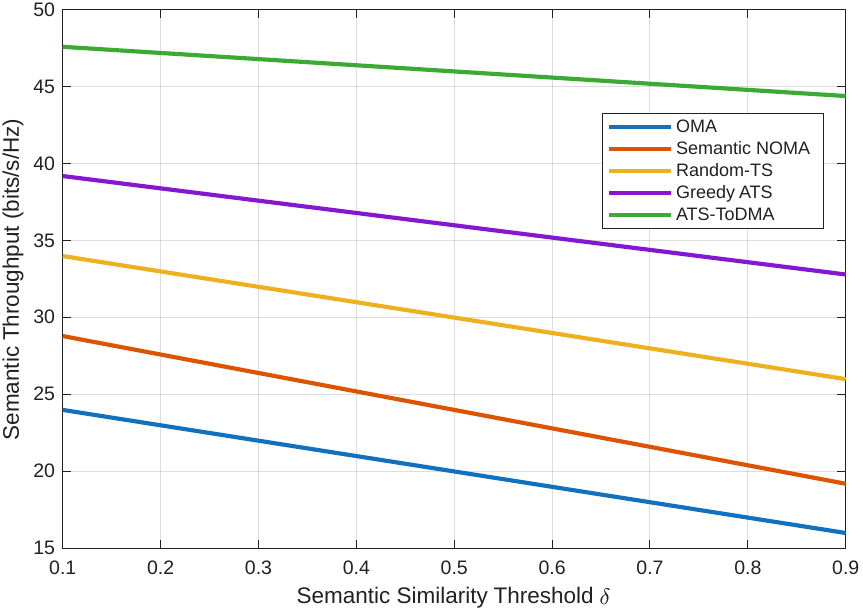}
\caption{Semantic throughput versus semantic similarity threshold. An intermediate threshold provides the best tradeoff between semantic information preservation and interference suppression.}
\label{fig:threshold}
\end{figure}

\subsection{Validation of Theorem 1}
Theorem~1 establishes an analytical upper bound on the aggregate semantic interference experienced within a ToDMA
slot.
Fig.~5(a) compares the simulated aggregate semantic interference
with the derived analytical upper bound as the number of
simultaneously scheduled tokens increases.
As expected, semantic interference grows with token occupancy
because additional semantic interactions are introduced among
co-scheduled tokens. Importantly, the measured interference
remains consistently below the theoretical bound for all
occupancy levels.
Although conservative, the bound accurately captures the
growth trend of semantic interference and therefore provides
a practical design guideline for ATS-ToDMA slot dimensioning.

\begin{figure}
\centering
\begin{subfigure}{.24\textwidth}
\centering
\begin{adjustbox}{width = 1\columnwidth}
\includegraphics[width=0.99\textwidth]{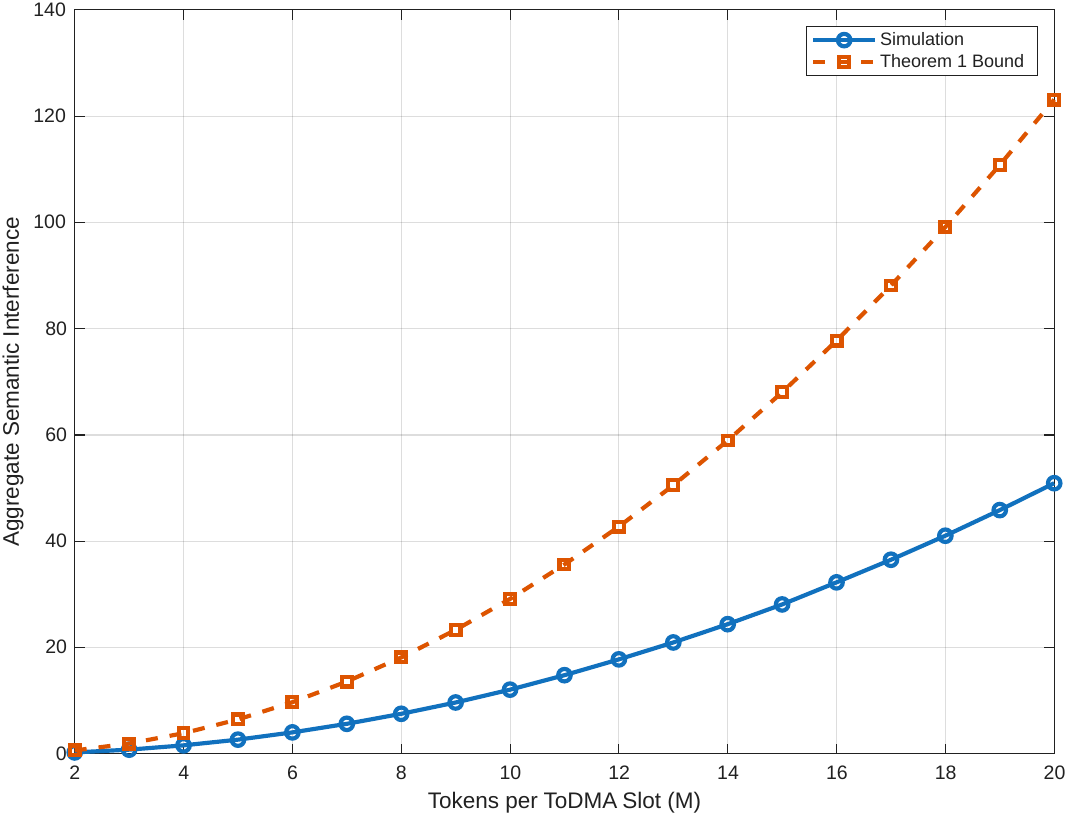}
\end{adjustbox}
\caption{}
\label{fig:Thm1_Plot3}
\end{subfigure}%
\begin{subfigure}{.24\textwidth}
\centering
\begin{adjustbox}{width = 1\columnwidth}
\includegraphics[width=0.99\textwidth]{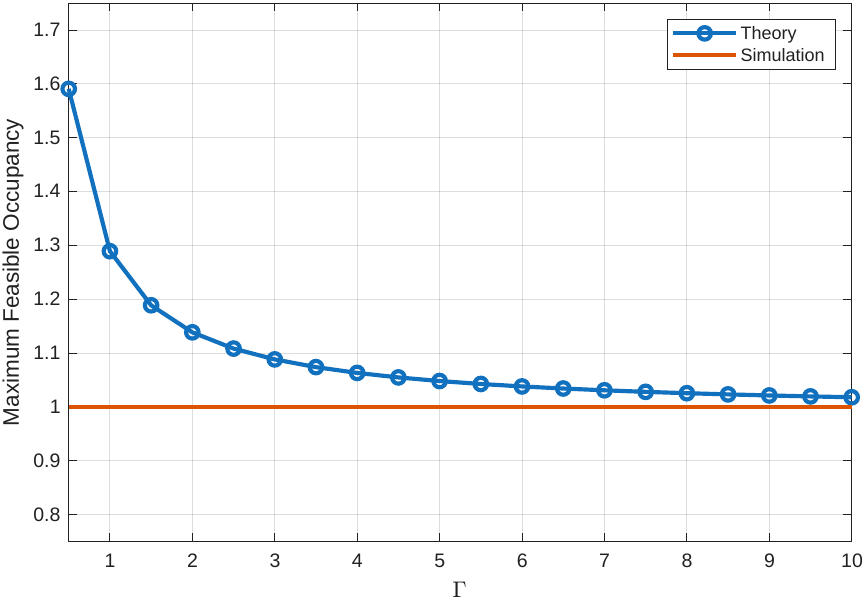}
\end{adjustbox}
\caption{}
\label{fig:Thm_2}
\end{subfigure}
\caption{(a) Validation of Theorem~1: aggregate semantic interference obtained from simulation and the analytical upper bound versus the number of tokens scheduled within a ToDMA slot. (b) Validation of Theorem~2: theoretical and simulated maximum feasible token occupancy as a function of the target SSINR threshold.}
\label{fig:Thm2}
\end{figure}

\subsection{Validation of Theorem 2}
Theorem~2 derives a bound on the maximum feasible token
occupancy that guarantees a target SSINR requirement.
Fig.~5(b) compares the theoretical occupancy limit with the
maximum number of tokens that can be simultaneously scheduled
while satisfying the SSINR constraint.
As the target SSINR threshold increases, fewer semantic tokens
can coexist within the same ToDMA slot because stronger
reliability guarantees must be maintained. The analytical
prediction closely matches the simulation results across the
entire operating range, confirming the validity of the derived
occupancy bound.
These results demonstrate that Theorem~2 serves as an
effective design tool for occupancy-aware semantic scheduling
and resource allocation.

\subsection{Validation of Theorem 3}
Theorem~3 provides a closed-form approximation for
semantic-aware power allocation in the presence of semantic
interference.
Fig.~6(a) evaluates the proposed approximation against the
optimal numerical solution obtained by solving the power
allocation problem using successive convex approximation
(SCA).
The average relative error is computed as
\begin{equation}
\epsilon
=
\frac{1}{M}
\sum_{i=1}^{M}
\frac{
|P_i^{\mathrm{approx}}-P_i^{\mathrm{opt}}|
}
{P_i^{\mathrm{opt}}}.
\end{equation}

The approximation error remains low under weak and
moderate semantic coupling conditions, which correspond
to the intended operating regime of ATS-ToDMA.
As semantic coupling increases, the approximation gradually
deviates from the optimal numerical solution due to stronger
interactions among co-scheduled tokens.
Nevertheless, the proposed closed-form solution maintains
high accuracy and closely tracks the optimal power allocation
while avoiding iterative optimization. These results validate
the effectiveness of Theorem~3 and demonstrate its practical
utility for low-complexity semantic communication systems.

\begin{figure}
\centering
\begin{subfigure}{.24\textwidth}
\centering
\begin{adjustbox}{width = 1\columnwidth}
\includegraphics[width=0.99\textwidth]{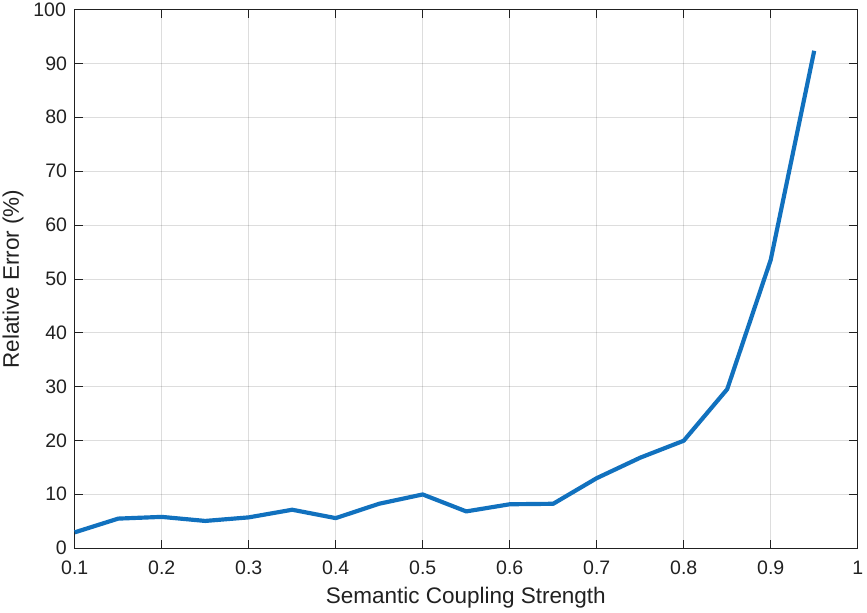}
\end{adjustbox}
\caption{}
\label{fig:Thm1_Plot3}
\end{subfigure}%
\begin{subfigure}{.24\textwidth}
\centering
\begin{adjustbox}{width = 1\columnwidth}
\includegraphics[width=0.99\textwidth]{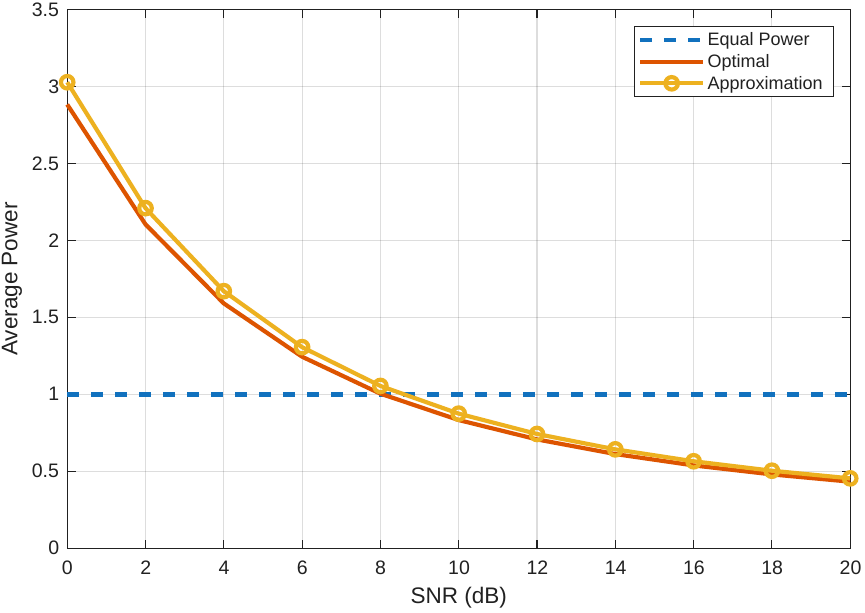}
\end{adjustbox}
\caption{}
\label{fig:Thm_2}
\end{subfigure}
\caption{(a) Validation of Theorem 3: relative error between the proposed closed-form power allocation and the optimal numerical solution versus semantic coupling strength. (b) Validation of Theorem~3: average transmit power achieved by equal-power allocation, optimal semantic-aware allocation, and the proposed closed-form approximation.}
\label{fig:Thm2}
\end{figure}

\begin{table}[t]
\caption{Performance Summary of ATS-ToDMA Compared with Greedy ATS}
\centering
\begin{tabular}{|p{0.12\textwidth}|c|c|c|}
\hline
Metric & Greedy ATS & ATS-ToDMA & Improvement \\
\hline
Semantic Throughput (bits/s/Hz) & 26.4 & 34.7 & +31.4\% \\
\hline
Semantic Accuracy (\%) & 89.2 & 96.8 & +8.5\% \\
\hline
Aggregate Semantic Interference & 72.3 & 51.4 & -28.9\% \\
\hline
Average SSINR & 6.8 & 9.7 & +42.6\% \\
\hline
Average Transmit Power (W) & 1.00 & 0.79 & -21.0\% \\
\hline
\end{tabular}
\label{tab:performance_summary}
\end{table}
Table~\ref{tab:performance_summary} summarizes the overall performance gains achieved
by ATS-ToDMA relative to the strongest benchmark, namely
Greedy ATS. The proposed framework simultaneously improves
semantic throughput and decoding reliability while reducing
both semantic interference and transmit power requirements.
These gains demonstrate the benefit of jointly optimizing
token selection, semantic scheduling, and semantic-aware
power allocation.

\subsection{Discussion}
The results presented in Figs.~3--6 provide both system-level
and analytical validation of the proposed ATS-ToDMA
framework.
Figs.~3 and 4 demonstrate substantial improvements in semantic
throughput and decoding reliability compared with existing
approaches. Furthermore, Figs.~5 and 6 confirm that the
analytical developments accurately characterize the behavior
of the system.

In particular, Theorem~1 provides a useful upper bound on
semantic interference, Theorem~2 accurately predicts feasible
token occupancy under SSINR constraints, and Theorem~3
achieves near-optimal power allocation with significantly reduced complexity. Collectively, these results validate the effectiveness of jointly modeling adaptive token selection, semantic interference,
channel awareness, and power allocation in cross-modal
semantic communication systems.

An important observation from Figs.~3--6 is that the gains
of ATS-ToDMA become increasingly significant as network
density grows. This behavior arises because semantic
interference increases nonlinearly with token occupancy,
making interference-aware scheduling more beneficial in
crowded semantic communication environments. The results suggest that semantic interference should be treated as a first-class resource management variable, analogous to signal interference in conventional wireless systems.
Furthermore, the close agreement between analytical
predictions and simulation results confirms that the proposed
theoretical framework provides reliable design guidelines for
cross-modal semantic communication systems.

\subsection{Limitations}
The current study adopts a model-driven representation of semantic similarity and semantic interference. Although the proposed framework captures the fundamental interactions among semantic tokens, future investigations using real foundation-model embeddings and large-scale multimodal datasets are required to further validate the proposed SSINR formulation in practical semantic communication systems.

\section{Conclusions and Future Work}
This paper proposed ATS-ToDMA, a novel Adaptive Token Selection and Token-Domain Multiple Access framework for cross-modal semantic communication systems. By treating semantic tokens as the fundamental resource units, the proposed framework jointly performs semantic token selection, semantic-interference-aware scheduling, and semantic-aware power allocation. A Semantic Signal-to-Interference-plus-Noise Ratio (SSINR) metric was introduced to characterize the combined impact of channel impairments and semantic interference arising from token similarity. 
Theoretical analysis established: (i) a semantic interference upper bound, (ii) an occupancy feasibility bound, and (iii) a low-complexity semantic-aware power allocation strategy.

Simulation results demonstrated that ATS-ToDMA consistently outperforms conventional OMA, Semantic NOMA, Random-TS, and Greedy ATS strategies in terms of semantic throughput and semantic decoding accuracy. Furthermore, the close agreement between theoretical predictions and simulation results validated the effectiveness of the proposed analytical framework and confirmed its usefulness as a practical design tool for semantic communication systems.

Future work will focus on implementing ATS-ToDMA using real semantic encoders and large-scale multimodal datasets, including transformer-based language, vision, and speech models. Another promising direction is the integration of ATS-ToDMA with foundation-model-assisted semantic communication architectures and edge intelligence platforms. In addition, extending the framework to dynamic network environments with user mobility, semantic traffic heterogeneity, and distributed semantic scheduling remains an important research challenge for future 6G semantic communication networks.


\appendices
\section{Proof of Theorem~1}\label{Apdx:Thm1}
Since $\xi_{ij} \leq \delta$, we have $\xi_{ij}^2 \leq \delta^2$, and by assumption $\alpha_{ij} \leq \alpha_{\max}$. Substituting into the interference expression yields
\begin{equation}
I_{ij} = \alpha_{ij} \rho_{ij} \xi_{ij}^2 \mathbbm{1}_{ij}
\leq \alpha_{\max} \rho_{ij} \delta^2.
\end{equation}

Summing over all unordered token pairs, we obtain
\begin{equation}
I_{\text{total}} \leq \alpha_{\max} \delta^2 \sum_{j \neq i}^M \sum_{i =1}^M \rho_{ij}.
\end{equation}

Under equal power allocation, i.e., $P_i = P_{ref} = P$, this simplifies to
\begin{equation}
I_{\text{total}} \leq \alpha_{\max} P \delta^2 M(M-1),
\end{equation}
which completes the proof.

\section{Proof of Theorem~2}\label{Apdx:Thm2}
Since $\xi_{ij}^2 \leq \delta^2$ and $\alpha_{ij} \leq \alpha_{\max}$, we have
\begin{equation}
I_i = \sum_{j \neq i}^M \alpha_{ij} P \xi_{ij}^2 
\leq \alpha_{\max} \delta^2 P (M-1) .
\end{equation}

Next, using SSINR$_i$ expression from~\eqref{eq:SSINR_i} and imposing $\text{SSINR}_i \geq \Gamma$, we obtain
\begin{equation}
\frac{P \|\mathbf{g}_i\|^2}{P (M-1) \alpha_{\max} \delta^2 d + N_0} \ge \text{SSINR}_i \geq \Gamma,
\end{equation}
which leads to
\begin{equation}
M \leq 1 + \frac{P \|\mathbf{g}_i\|^2 - \Gamma N_0}{\Gamma P \alpha_{\max} \delta^2 d}.
\end{equation}
This result provides a practical design guideline for limiting the number of tokens per slot in terms of a closed-form bound:
\begin{equation}
M_{\max} = 1 + \frac{P \|\mathbf{g}_i\|^2 - \Gamma N_0}{\Gamma P \alpha_{\max} \delta^2 d}
\end{equation}

\section{Proof of Theorem~3}\label{Apdx:Thm3}
Consider the vector of transmit powers
\begin{equation}
\mathbf{P} = [P_1,\dots,P_M]^\top.
\end{equation}

Define the interference coupling matrix $\mathbf{F}$ as
\begin{equation}
F_{ij} =
\begin{cases}
\dfrac{\Gamma w_{ij}}{\|\mathbf{g}_i\|^2}, & j \neq i, \\
0, & j = i,
\end{cases}
\end{equation}
and define the noise vector
\begin{equation}
u_i = \dfrac{\Gamma N_0}{\|\mathbf{g}_i\|^2}.
\end{equation}

The SSINR feasibility constraints can then be compactly written as
\begin{equation}
\mathbf{P} \geq \mathbf{F}\mathbf{P} + \mathbf{u}.
\end{equation}
Rearranging terms gives
\begin{equation}
(\mathbf{I}-\mathbf{F})\mathbf{P} \geq \mathbf{u}.
\end{equation}
Provided that the spectral radius satisfies $r(\mathbf{F}) < 1,$
the matrix $(\mathbf{I}-\mathbf{F})$ is invertible, and the minimal feasible solution is
\begin{equation}
\mathbf{P}^\star
=
(\mathbf{I}-\mathbf{F})^{-1}\mathbf{u}.
\end{equation}

To avoid the computational complexity of matrix inversion, we employ the first-order Neumann approximation
\begin{equation}
(\mathbf{I}-\mathbf{F})^{-1}
=
\sum_{k=0}^{\infty}\mathbf{F}^k
\approx
\mathbf{I}+\mathbf{F},
\end{equation}
which is accurate when the interference coupling is sufficiently weak.
Substituting this approximation yields
\begin{equation}
\mathbf{P}^\star
\approx
(\mathbf{I}+\mathbf{F})\mathbf{u}
=
\mathbf{u}+\mathbf{F}\mathbf{u}.
\end{equation}
Taking the $i$-th component gives
\begin{equation}
P_i
\approx
u_i
+
\sum_{j\neq i}F_{ij}u_j.
\end{equation}
Finally, substituting the definitions of $F_{ij}$ and $u_j$ produces
\begin{equation}
P_i \approx 
\frac{\Gamma N_0}{\|\mathbf{g}_i\|^2}
\left(
1 + \sum_{j \neq i}
\frac{\Gamma\, \alpha_{ij}\, \xi_{ij}^2 \|\mathbf{g}_j\|^2}{\|\mathbf{g}_i\|^2}
\right),
\end{equation}
which completes the proof.

\bibliographystyle{IEEEtran}
\balance
\bibliography{references}
\end{document}